\documentclass[copyright]{eptcs}
\providecommand{\event}{J.~Dassow, G.~Pighizzini, B.~Truthe (Eds.): 11th International Workshop\\
on Descriptional Complexity of Formal Systems (DCFS 2009)} 
\providecommand{\volume}{3}
\providecommand{\anno}{2009}
\providecommand{\firstpage}{189} 
\providecommand{\eid}{18}
\usepackage{breakurl}        

\input{dcfs.tex}

\newcommand{\Ra}{\Rightarrow}
\newcommand{\eps}{\lambda}

\begin{document}
\title{Descriptional Complexity of Three-Nonterminal Scattered Context Grammars: An Improvement}
\def\titlerunning{Three-Nonterminal Scattered Context Grammars: An Improvement}
\author{Tom{\' a}{\v s} Masopust \qquad Alexander Meduna
\institute{Faculty of Information Technology -- Brno University of Technology\\
  Bo\v{z}et\v{e}chova 2 -- Brno 61266 -- Czech Republic}
\email{\{masopust,meduna\}@fit.vutbr.cz}
}
\def\authorrunning{T.~Masopust, A.~Meduna}

\maketitle
  \begin{abstract}
    Recently, it has been shown that every recursively enumerable language can be generated by a scattered context grammar with no more than three nonterminals. However, in that construction, the maximal number of nonterminals simultaneously rewritten during a derivation step depends on many factors, such as the cardinality of the alphabet of the generated language and the structure of the generated language itself. This paper improves the result by showing that the maximal number of nonterminals simultaneously rewritten during any derivation step can be limited by a small constant regardless of other factors.
  \end{abstract}

\section{Introduction}
  Scattered context grammars, introduced by Greibach and Hopcroft in~\cite{GreHop}, are partially parallel rewriting devices based on context-free productions, where in each derivation step, a finite number of nonterminal symbols of the current sentential form is simultaneously rewritten. As scattered context grammars were originally defined without erasing productions, it is no surprise that they generate only context sensitive languages. On the other hand, however, the question of whether every context sensitive language can be generated by a (nonerasing) scattered context grammar is an interesting, longstanding open problem. Note that the natural generalization of these grammars allowing erasing productions makes them computationally complete (see \cite{meduna:eatcs}). For some conditions when a scattered context grammar can be transformed to an equivalent nonerasing scattered context grammar, the reader is referred to \cite{techetAI}. In what follows, we implicitly consider scattered context grammars with erasing productions.

  Although many interesting results have been achieved in the area of the descriptional complexity of scattered context grammars during the last few decades, the main motivation to re-open this investigation area comes from an interesting, recently started research project on bulding parsers and compilers of programming languages making use of advantages of scattered context grammars (see, for instance, papers \cite{kolar,rychnov} for more information on the advantages and problems arising from this approach).

  To give an insight into the descriptional complexity of scattered context grammars (including erasing productions), note that it is proved in \cite{meduna2} that one\mbox{-}nonterminal scattered context grammars are not powerful enough to generate all context sensitive languages so that it is demonsrated that they are not able to generate the language $\{a^{2^{2^n}} : n\ge 0\}$ (which is scattered context, see Lemma \ref{lem1} below). In addition, although they are not able to generate all these languages, it is an open problem (because of the erasing productions) whether they can generate a language which is not context sensitive. On the other hand, it is proved in \cite{Meduna00b} that three nonterminals are sufficient enough for scattered context grammars to characterize the family of recursively enumerable languages. In that proof, however, the maximal number of nonterminal symbols simultaneously rewritten during any derivation step depends on the alphabet of the generated language and on the structure of the generated language itself.

  Later, in \cite{vaszil}, Vaszil gave another construction limiting the maximal number of nonterminals simultaneously rewritten during one derivation step. However, this improvement is for the price of increasing the number of nonterminals. Although Vaszil's construction has been improved since then (in the sense of the number of nonterminals, see \cite{masopustTCS} for an overview of the latest results), the number of three nonterminals has not been achieved.

  This paper presents a construction improving the descriptional complexity of scattered context grammars with three nonterminals by limiting the maximal number of nonterminals simultaneously rewritten during any derivation step regardless of any other factors. This result is achieved by the combination of approaches of both previously mentioned papers. Specifically, this paper proves that every recursively enumerable language is generated by a three-nonterminal scattered context grammar, where no more than nine symbols are simultaneously rewritten during any derivation step. This is a significant improvement in comparison with the result of \cite{Meduna00b}, where more than $2n+4$ symbols have to be simultaneously rewritten during almost all derivation steps of any successful derivation, for some $n$ strictly greater than the number of terminal symbols of the generated language plus two. To be more precise, $n$ strongly depends not only on the terminal alphabet of the generated language, but also on the structure of the generated language itself.

  Finally, note that analogously as in \cite{Meduna00b}, we do not give a constant limit on the number of non-context-free productions, which is also limited by fixed constants in \cite{vaszil} and \cite{masopustTCS}. To find such a limit is an interesting challenge for the future research, as well as to find out whether the number of nonterminals can be reduced to two. See also the overview of known results and open problems in the conclusion.

\section{Preliminaries and definitions}
  We assume that the reader is familiar with formal language theory (see \cite{salomaa}).
  For an alphabet (finite nonempty set) $V$, $V^*$ represents the free monoid generated 
  by~$V$ with the unit denoted by $\eps$. Set $V^+ = V^*-\{\eps\}$. For $w \in V^*$ and $a\in V$, let $|w|$, $|w|_a$, and $w^R$ denote the length of $w$, the number of occurrences of $a$ in $w$, and the mirror image of $w$, respectively.

  A {\em scattered context grammar\/} is a quadruple $G=(N,T,P,S)$, where $N$ is the alphabet of nonterminals, $T$ is the alphabet of terminals such that $N\cap T=\emptyset$, $S\in N$ is the start symbol, and $P$ is a finite set of productions of the form $(A_1,A_2,\dots,A_n)\to (x_1,x_2,\dots,x_n)$, for some $n\geq 1$, where $A_i\in N$ and $x_i \in (N\cup T)^*$, for all $i=1,\dots,n$. If $n\ge 2$, then the production is said to be {\em non\mbox{-}context-free}; otherwise, it is context-free. In addition, if for each $i=1,\dots,n$, $x_i\neq\eps$, then the production is said to be {\em nonerasing}; $G$ is {\em nonerasing\/} if all its productions are nonerasing.

  For $u,v\in (N\cup T)^*$, $u\Ra v$ in $G$ provided that
  \begin{itemize}
    \item $u=u_1A_1u_2A_2u_3\dots u_nA_nu_{n+1}$,
    \item $v=u_1x_1u_2x_2u_3\dots u_nx_nu_{n+1}$, and
    \item $(A_1,A_2,\dots,A_n)\to (x_1,x_2,\dots,x_n)\in P$,
  \end{itemize}
  where $u_i \in (N\cup T)^*$, for all $i=1,\dots,n+1$. The language generated by $G$ is 
  defined as 
  $$L(G)=\{w \in T^* : S \Ra^* w\},$$ 
  where $\Ra^*$ denotes the reflexive and 
  transitive closure of the relation $\Ra$. A language $L$ is said to be 
  a (nonerasing) {\em scattered context language} if there is a (nonerasing) scattered 
  context grammar~$G$ such that $L=L(G)$.

\section{Main results}
  First, we give a simple example of a nonerasing scattered context grammar generating a non\mbox{-}context-free language. Then, we present a nonerasing scattered context grammar generating the nontrivial context sensitive language $\{a^{l^{k^n}} : n\ge 0\}$, for any $k,l\ge 2$. Thus, for $k=l=2$, we have a scattered context grammar generating the language mentioned in the introduction. Note that independently on $k$ and $l$, the grammar has only twelve nonterminals and fourteen productions, ten of which are non-context-free.
  \begin{example}
    Let $G=(\{S,A,B,C\},\{a,b,c\},P,S)$ be a scattered context grammar with $P$ containing the following productions
    \begin{itemize}
      \item $(S)\to(ABC)$
      \item $(A,B,C)\to(aA,bB,cC)$
      \item $(A,B,C)\to(a,b,c)$
    \end{itemize}
    Then, it is not hard to see that the language generated by $G$ is 
    \[L(G)=\{a^nb^nc^n : n\ge 1\}.\tag*{$\diamond$}\]
  \end{example}

  \begin{lemma}\label{lem1}
    For any $k,l\ge 2$, the language $\{a^{l^{k^n}} : n\ge 0\}$ is a nonerasing scattered context language.
  \end{lemma}
  \begin{proof}
    Let $G=(\{S,A,A',A'',B,C,X,X_2,X_3,Y,Z,Z'\},\{a\},P,S)$ be a nonerasing scattered context grammar with $P$ containing the following productions:
    \begin{enumerate}
      \item\label{b1} $(S) \to (a^l)$,
      \item\label{b2} $(S) \to \big(a^{l^k}\big)$,
      \item\label{b3} $(S) \to \Big(a^{l^{k^2}}\Big)$,
      \item\label{b4} $(S) \to (A''A^{l-1}X_2B^{k^2-3}A'C^{k^2-1}XY)$,
      \item[*] first stage
      \item\label{b5} $(A',C,X,Y) \to (B^{k-1},A',X,C^{k}Y)$,
      \item\label{b6} $(A',X,Y) \to (B^{k-1},A',C^{k-1}XY)$,
      \item\label{b7} $(A',X,Y) \to (Z,Z,Y)$,
      \item\label{b8} $(Z,C,Z,Y) \to (Z,B^{k-1},Z,Y)$,
      \item\label{b9} $(Z,Z,Y) \to (B,B^{k-1},X_3)$,
      \item[*] second stage
      \item\label{b10} $(A'',A,X_2,X_3) \to (a^{l-1},A'',X_2A^{l},X_3)$,
      \item\label{b11} $(A'',X_2,B,X_3) \to (a^{l-1},A'',A^{l-1}X_2,X_3)$,
      \item\label{b12} $(A'',X_2,X_3) \to (Z',Z',X_3)$,
      \item\label{b13} $(Z',A,Z',X_3) \to (Z',a^{l-1},Z',X_3)$,
      \item\label{b14} $(Z',Z',X_3) \to (a,a^{l-1},a^{l-1})$.
    \end{enumerate}
    Then, all the possible successful derivations of $G$ are summarized in the following (strings in the square brackets are regular expressions describing the productions applied during the derivations).
    \[\begin{array}{rllcl}
      S & \Ra   & a^l                                     & \qquad & [\textrm{(\ref{b1})}]\\
      S & \Ra   & a^{l^k}                                 & \qquad & [\textrm{(\ref{b2})}]\\
      S & \Ra   & a^{l^{k^2}}                             & \qquad & [\textrm{(\ref{b3})}]\\
      S & \Ra   & A''A^{l-1}X_2B^{k^2-3}A'C^{k^2-1}XY     & & [\textrm{(\ref{b4})}]\\
        & \Ra^* & A''A^{l-1}X_2B^{k^n-2}X_3               & & [\textrm{((\ref{b5})$^+$(\ref{b6}))$^*$(\ref{b7})(\ref{b8})$^+$(\ref{b9})}]\\
        & \Ra^* & a^{l^{k^n-1}-l}A''A^{l^{k^n-1}-1}X_2X_3 & & [\textrm{((\ref{b10})$^+$(\ref{b11}))$^*$}]\\
        & \Ra^* & a^{l^{k^n}}                             & & [\textrm{(\ref{b12})(\ref{b13})$^+$(\ref{b14})}]\,.
    \end{array}\]

    The first three cases are clear. In the last case, $l$ symbols $A$ (including~$A''$) are generated in the first derivation step. Then, the derivation can be divided into two parts: in the first part, only productions from the first stage are applied (because there is no $X_3$ in the sentential form) generating $k^n$ auxiliary symbols ($B$'s, $X_2$, and $X_3$). Then, in the second part, only productions from the second stage are applied (because there is no $Y$ in the sentential form) generating $l^{k^n}$ symbols $a$. More precisely, we prove that all sentential forms of a successful derivation containing $X_3$, i.\,e. of the second part, are of the form \[a^{l^{m-1}-l}A''A^{l^{m-1}-1}X_2B^{k^n-m}X_3\,,\] for all $m=2,3,\dots,k^n$ and $n\ge 3$. Clearly, for $m=2$, the sentential form is $A''A^{l-1}X_2B^{k^n-2}X_3$. For $m=k^n$, we have $a^{l^{k^n-1}-l}A''A^{l^{k^n-1}-1}X_2X_3$ and it is not hard to prove that
    \begin{eqnarray*}
      a^{l^{k^n-1}-l}A''A^{l^{k^n-1}-1}X_2X_3 & \Ra^* & a^{l^{k^n-1}-l}aa^{(l-1)(l^{k^n-1}-1)}a^{l-1}a^{l-1} = a^{l^{k^n}}\,.
    \end{eqnarray*} Thus, assume that $2\le m < k^n$. Then,
    \begin{eqnarray*}
      &       & a^{l^{m-1}-l}A''A^{l^{m-1}-1}X_2B^{k^n-m}X_3\\
      & \Ra^* & a^{l^{m-1}-l}a^{(l-1)(l^{m-1}-1)}A''X_2A^{l(l^{m-1}-1)}B^{k^n-m}X_3\quad [\textrm{(\ref{b10})}^*]\\
      &&=      a^{l^m-2l+1}A''X_2A^{l^m-l}B^{k^n-m}X_3\\
      & \Ra   & a^{l^m-2l+1}a^{l-1}A''A^{l^m-l}A^{l-1}X_2B^{k^n-m-1}X_3\quad [\textrm{(\ref{b11})}]\\
      &&=     a^{l^m-l}A''A^{l^m-1}X_2B^{k^n-(m+1)}X_3\,.
    \end{eqnarray*}

    For a complete proof of the correctness of this construction, the reader is referred to \cite{masopust:mono}.
  \end{proof}

  Now, we prove the main result of this paper.

  \begin{theorem}
    Every recursively enumerable language is generated by a scattered context grammar with three nonterminals, where no more than nine nonterminals are simultaneously rewritten during one derivation step.
  \end{theorem}
  \begin{proof}
    Let $L$ be a recursively enumerable language. Then, by Geffert~\cite{Geffert:1988}, 
    there is a grammar\linebreak \hbox{$G'=(\{S',A,B,C,D\},T,P\cup\{AB\to\eps,CD\to\eps\},$ $S')$}, 
    where $P$ contains only context-free productions of the following three forms: 
    $S'\to uS'a$, $S'\to uS'v$, $S'\to \eps$, for $u\in\{A,C\}^*$, $v\in\{B,D\}^*$, 
    and $a\in T$. In addition, Geffert proved that any successful derivation of $G'$ is 
    divided into two parts: the first part is of the form 
    \[S' \Ra^* w_1S'w_2w \Ra w_1w_2w\,,\] 
    generated only by context-free productions from $P$, where $w_1\in\{A,C\}^*$, 
    $w_2\in\{B,D\}^*$, and $w\in T^*$, and the second part is of the form 
    \[w_1w_2w\Ra^* w\,,\] 
    generated only by productions $AB\to\eps$ and $CD\to\eps$.

    Let $G=(\{S,A,B\},T,P,S)$ be a scattered context grammar with $P$ constructed as follows:
    \begin{enumerate}
      \item\label{genS} $(S) \to (SBBASABBSA)$,
      \item\label{a2} $(S,S,S) \to (S,h(u)Sh(a),S)$ \qquad if $S' \to uS'a \in P'$,
      \item\label{a3} $(S,S,S) \to (S,h(u)Sh(v),S)$ \qquad if $S' \to uS'v \in P'$,
      \item\label{a4} $(S,A,B,B,S,B,B,A,S) \to (\eps,\eps,\eps,S,S,S,\eps,\eps,\eps)$,
      \item\label{a5} $(S,B,A,B,S,B,A,B,S) \to (\eps,\eps,\eps,S,S,S,\eps,\eps,\eps)$,
      \item\label{a6} $(S,B,B,A,S,A,B,B,S) \to (\eps,\eps,\eps,SBBA,S,S,\eps,\eps,\eps)$,
      \item\label{a7} $(S,B,B,A,S,A,B,B,S) \to (\eps,\eps,\eps,S,S,S,\eps,\eps,\eps)$,
      \item\label{remS} $(S,S,S,A) \to (\eps,\eps,\eps,\eps)$,
    \end{enumerate}
    where $h$ is a homomorphism from $(\{A,B,C,D\}\cup T)^*$ to $(\{A,B\}\cup T)^*$ defined as $h(A)=ABB$, $h(B)=BBA$, $h(C)=h(D)=BAB$, and $h(a)=AaBB$, for all $a\in T$.

    To prove that $L(G')\subseteq L(G)$, consider a successful derivation of $w\in T^*$ 
    in~$G'$. Such a derivation is of the form described above, where the second part of the derivation is according to a sequence $p_1p_2\dots p_r$ of productions $AB\to\eps$ and $CD\to\eps$, for some $r\ge 0$. Then, in $G$, the derivation of $w$ can be simulated by applications of the corresponding productions constructed above as follows:
    \begin{eqnarray*}
      S & \Ra   & SBBASABBSA \quad [\textrm{(\ref{genS})}]\\
        & \Ra^* & SBBAh(w_1)Sh(w_2w)ABBSA \quad [\textrm{(\ref{a2})$^*$(\ref{a3})$^*$}]\\
        & \Ra^* & Sh(w_1)Sh(w_2)SwA \quad [\textrm{(\ref{a6})$^*$(\ref{a7})}]\\
        & \Ra^* & SSSwA \quad [q_r\dots q_2q_1]\\
        & \Ra   & w \quad [\textrm{(\ref{remS})}]\,,
    \end{eqnarray*}
    where, for each $1\le i\le r$,
    $$q_i=\begin{cases}
    (S,A,B,B,S,B,B,A,S)\to (\eps,\eps,\eps,S,S,S,\eps,\eps,\eps), & \text{ if $p_i=AB\to\eps$}, \\
    (S,B,A,B,S,B,A,B,S)\to (\eps,\eps,\eps,S,S,S,\eps,\eps,\eps), & \text{ otherwise}.
    \end{cases}$$

    On the other hand, to prove that $L(G)\subseteq L(G')$, we demonsrate that $G'$ generates any $x\in L(G)$.

    First, we prove that each of the productions (\ref{genS}) and (\ref{remS}) is applied exactly once in each successful derivation of $G$. To prove this, let $S\Ra^* x$ be a derivation of a string $x\in(\{S,A,B\}\cup T)^*$. Let $i$ be the number of applications of production~(\ref{genS}), $j$ be the number of applications of production (\ref{remS}), and $2k$ be the number of $B$'s in $x$. Then, it is not hard to see that
    \begin{itemize}
      \item $|x|_B = 2k$,
      \item $|x|_A = k + i - j$,
      \item $|x|_S = 1 + 2i - 3j$.
    \end{itemize}
    Thus, for $x\in T^*$, we have that $2k=0$ and $i=j$. In addition, $1+2i-3i=0$ implies that $i=1$, which means that each of the productions (\ref{genS}) and (\ref{remS}) is applied exactly once in each successful derivation of $G$---production (\ref{genS}) as the first production and production (\ref{remS}) as the last production of the derivation. We have shown that every successful derivation of $G$ is of the form \[S\Ra SBBASABBSA \Ra^* w_1Sw_2Sw_3Sw_4A \Ra w_1w_2w_3w_4\,,\] for some terminal strings $w_1,w_2,w_3,w_4\in T^*$.

    Furthermore, there is no production that can change the position of the middle 
    symbol $S$. Therefore, with respect to productions of $G$, we have 
    that
    $w_1, w_2\in \{A,B\}^*$, which along with $w_1,w_2\in T^*$ implies that $w_1=w_2=\eps$. Thus, the previously shown successful derivation is of the form \[S\Ra SBBASABBSA \Ra^* SSw_3Sw_4A \Ra w_3w_4\,.\] Analogously, it can be seen that $w_3\in\{BAB,BBA,AaBB : a\in T\}^*$. Therefore, from the same reason as above, $w_3=\eps$ and every successful derivation of $G$ is of the form
    \begin{eqnarray}
      S \Ra SBBASABBSA \Ra^* SSSwA \Ra w\,,
    \end{eqnarray}
    for some $w\in T^*$.

    Consider any inner sentential form of a successful derivation of $G$. Such a sentential form is a string of the form \[u_1Su_2Su_3Su_4A\,,\] for some $u_i\in (\{A,B\}\cup T)^*$, $1\le i\le 4$. However, it is not hard to see that $u_1=\eps$ and $u_4\in T^*$; otherwise, if there is a nonterminal symbol appearing in the string $u_1u_4$, then, according to the form of productions, none of these symbols can be removed and, therefore, the derivation cannot be successful. Thus, every inner sentential form of any successful derivation of $G$ is of the form
    \begin{eqnarray}\label{sf}
      S\bar{u}S\bar{v}S\bar{w}A\,,
    \end{eqnarray}
    where $\bar{u}\in(BBA+\eps)\{ABB,BAB\}^*$, $\bar{v}\in\{BAB,BBA,AaBB : a\in T\}^*$, and $\bar{w}\in T^*$. Now, we prove that \[\bar{v}\in\{BBA,BAB\}^*\{AaBB : a\in T\}^*(ABB+\eps)\,.\] In other words, we prove that any applications of productions (\ref{a6}) and (\ref{a7}) precede the first application of any of productions (\ref{a4}) and (\ref{a5}).

    Thus, consider the beginning of a successful derivation of the form \[S\Ra SBBASABBSA \Ra^* SBBAuSvABBSwA\,,\] where none of productions (\ref{a6}) and (\ref{a7}) has been applied, and the first application of one of these productions follows. Note that during this derivation, only productions (\ref{genS}) to (\ref{a3}) have been applied because the application of production (\ref{a4}) or (\ref{a5}) skips some nonterminal symbols and, therefore, leads to an incorrect sentential form (see the correct form (\ref{sf}) above). Clearly, $w=\eps\in T^*$ (it is presented here for the reason of induction).

    If production (\ref{a6}) follows, the derivation proceeds
    \begin{eqnarray}
      SBBAuSvABBSwA & \Ra & SBBAuSvSwA\,,
    \end{eqnarray}
    and if production (\ref{a7}) follows, the derivation proceeds
    \begin{eqnarray}
      SBBAuSvABBSwA & \Ra & SuSvSwA\,.
    \end{eqnarray}
    In addition, $w\in T^*$ and, according to the form of productions (\ref{genS}) 
    to (\ref{a3}),
    $u\in\{ABB,BAB\}^*$ and $v\in\{BBA,BAB,AaBB : a\in T\}^*$.

    Now, productions (\ref{a2}) and (\ref{a3}) can be applied. Let
    \begin{eqnarray}\label{5}
      SBBAuSvSwA & \Ra^* & SBBAuu_1Sv_1vSwA \quad [\textrm{((\ref{a2})+(\ref{a3}))}^*]
    \end{eqnarray}
    and
    \begin{eqnarray}\label{6}
      SuSvSwA & \Ra^* & Suu_1Sv_1vSwA \quad [\textrm{((\ref{a2})+(\ref{a3}))}^*]
    \end{eqnarray}
    be the longest parts of the derivation by productions (\ref{a2}) and (\ref{a3}), i.\,e., the application of one of productions (\ref{a4}) to (\ref{remS}) follows.

    ${\bf I.}$ In the first case, derivation $(\ref{5})$, each of productions (\ref{a4}), (\ref{a5}), and (\ref{remS}) leads to an incorrect sentential form. Thus, either production (\ref{a6}) or (\ref{a7}) has to be applied. In both cases, however, $v_1v$ has to be of the form $v'AaBB$, for some $a\in T$, i.\,e.,
    \begin{eqnarray}
      SBBAuu_1Sv'AaBBSwA & \Ra & SBBAu'Sv'SawA \quad [\textrm{(\ref{a6})}]
    \end{eqnarray}
    and the derivation proceeds as in $(\ref{5})$ or
    \begin{eqnarray}
      SBBAuu_1Sv'AaBBSwA & \Ra & Su'Sv'SawA \quad [\textrm{(\ref{a7})}]
    \end{eqnarray}
    and the derivation proceeds as in $(\ref{6})$ because 
    $$u'=uu_1\in\{ABB,BAB\}^* \mbox{ and }v'\in\{BBA,BAB,AaBB : a\in T\}^*.$$ 
    By induction,
    \begin{eqnarray}
      SBBAu'Sv'SawA & \Ra^* & Su''Sv''Sw''awA \quad [\textrm{((\ref{a2})+(\ref{a3})+(\ref{a6}))$^*$(\ref{a7})}]\,,
    \end{eqnarray}
    for some $u''\in\{ABB,BAB\}^*$, $v''\in\{BBA,BAB,AaBB : a\in T\}^*$, and
    $w''aw\in T^*$.

    ${\bf II.}$ In the second case, derivation $(\ref{6})$, each of productions (\ref{a6}) 
    and (\ref{a7}) leads to an incorrect sentential form, and production (\ref{remS}) 
    finishes the derivation, which, as shown above, implies that \hbox{$uu_1=v_1v=\eps$}. Thus, 
    assume that either production (\ref{a4}) or production (\ref{a5}) is applied. Then, in 
    the former case,\linebreak
    $uu_1=ABBu'$ and $v_1v=v'BBA$, and, in the latter case, $uu_1=BABu'$ and $v_1v=v'BAB$, i.\,e.,
    \begin{eqnarray}\label{9}
      SABBu'Sv'BBASwA & \Ra & Su'Sv'SwA \quad [\textrm{(\ref{a4})}]
    \end{eqnarray}
    and the derivation proceeds as in $(\ref{6})$ or
    \begin{eqnarray}\label{10}
      SBABu'Sv'BABSwA & \Ra & Su'Sv'SwA \quad [\textrm{(\ref{a5})}]
    \end{eqnarray}
    and the derivation also proceeds as in $(\ref{6})$ because 
    $$u'\in\{ABB,BAB\}^* \mbox{ and } v'\in\{BBA,BAB,AaBB : a\in T\}^*.$$
    Notice that the application of a production 
    constructed in (\ref{a2}) would lead, in its consequence, to an incorrect sentential 
    form because the derivation would reach one of the following two forms 
    $$SABBxSyAaBBSzA \ \mbox{ or }\ SBABxSyAaBBSzA,$$
    and each of productions (\ref{a4}) and~(\ref{a5}) would move 
    either $A$ in front of the first $S$, or at least one $B$ behind the last~$S$.  
    By induction, it implies that the successful derivation proceeds as
    \begin{eqnarray}
      Su'Sv'SwA & \Ra^* & SSSwA \Ra w \quad [\textrm{((\ref{a3})+(\ref{a4})+(\ref{a5}))$^*$(\ref{remS})}]\,.
    \end{eqnarray}

    Thus, we have proved that the following sequence of productions \[\textrm{((\ref{a4})+(\ref{a5}))((\ref{a2})+(\ref{a3}))$^*$((\ref{a6})+(\ref{a7}))}\] cannot be applied in any successful derivation of $G$. Therefore, all applications of productions (\ref{a6}) 
    and~(\ref{a7}) precede any application of productions (\ref{a4}) and (\ref{a5}), which means that \[\bar{v}\in\{BBA,BAB\}^*\{AaBB : a\in T\}^*(ABB+\eps)\,.\]

    Finally, by skipping all productions (\ref{a4}) and (\ref{a5}) in the considered successful derivation $S\Ra^* w$, we have
    \begin{eqnarray*}
      S & \Ra   & SBBASABBSA \quad [\textrm{(\ref{genS})}]\\
        & \Ra^* & SuSvSwA \quad [\textrm{((\ref{a2})+(\ref{a3})+(\ref{a6}))$^*$(\ref{a7})(\ref{a3})$^*$}]\\
        & \Ra   & uvw \quad [\textrm{(\ref{remS})}]\,,
    \end{eqnarray*}
    where $u\in\{ABB,BAB\}^*$, $v\in\{BBA,BAB\}^*$, $u=v^R$ (see~${\bf II}$), and $w\in T^*$. It is not hard to see that by applications of the corresponding productions constructed in (\ref{a2}) and (\ref{a3}), ignoring productions (\ref{a6}) and (\ref{a7}), and applying $S'\to\eps$ immediately after the last application of productions constructed in~(\ref{a3}), we have that $S'\Ra^* w_1w_2w$ in $G'$, where $w_1\in\{A,C\}^*$ and $w_2\in\{B,D\}^*$ are such that $h(w_1)=u$ and $h(w_2)=v$. As $u=v^R$, we have that $w_1w_2w\Ra^* w$ by productions $AB\to\eps$ and $CD\to\eps$, which completes the proof.
  \end{proof}

\section{Conclusion}
  This section summarizes the results and open problems concerning the descriptional complexity of scattered context grammars known so far.

  {\bf One-nonterminal scattered context grammars:}
  It is proved in \cite{meduna2} that scattered context grammars with only one nonterminal (including erasing productions) are not able to generate all context sensitive languages. However, because of the erasing productions, it is an open problem whether they can generate a language which is not context sensitive.

  {\bf Two-nonterminal scattered context grammars:}
  As far as the authors know, there is no published study concerning the generative power of scattered context grammars with two nonterminals.

  {\bf Three-nonterminal scattered context grammars:}
  In this paper, we have shown that scattered context grammars with three nonterminals, where no more than nine nonterminals are simultaneously rewritten during any derivation step, characterize the family of recursively enumerable languages. However, no other descriptional complexity measures, such as the number of non-context-free productions, are limited in this paper.

  Note that Greibach and Hopcroft \cite{GreHop} have shown that every scattered context grammar can be transformed to an equivalent scattered context grammar where no more than two nonterminals are simultaneously rewritten during any derivation step. This transformation, however, introduces many new nonterminals and, therefore, does not improve our result. Thus, it is an open problem whether the maximal number of nonterminals simultaneously rewritten during any derivation step can be reduced to two in case of scattered context grammars with three nonterminals.

  Finally, it is also an open problem whether the number of non-context-free productions can be limited.

  {\bf Four-nonterminal scattered context grammars:}
  It is proved in \cite{masopustTCS} that every recursively enumerable language can be generated by a scattered context grammar with four nonterminals and three non-context-free productions, where no more than six nonterminals are simultaneously rewritten during any derivation step. In comparison with the result of this paper, that result improves the maximal number of simultaneously rewritten symbols and limits the number of non-context-free productions. On the other hand, however, it requires more nonterminals.

  {\bf Five-nonterminal scattered context grammars:}
  It is proved in \cite{vaszil} that every recursively enumerable language can be generated by a scattered context grammar with five nonterminals and two non-context-free productions, where no more than four nonterminals are simultaneously rewritten during any derivation step. Note that this is the best known bound on the number of non-context-free productions. It is an interesting open problem whether this bound can also be achieved in case of scattered context grammars with three nonterminals.

  {\bf Scattered context grammars with one non-context-free production:}
  In comparison with the previous result, it is a natural question to ask what is the generative power of scattered context grammars with only one non-context-free production. However, as far as the authors know, this is another very interesting open problem.

  {\bf Nonerasing scattered context grammars:}
  So far, we have only considered scattered context grammars with erasing productions. However, the most interesting open problem in this investigation area is the question of what is the generative power of nonerasing scattered context grammars. It is not hard to see that they can generate only context sensitive languages. However, it is not known whether nonerasing scattered context grammars are powerful enough to characterize the family of context sensitive languages.

  Finally, from the descriptional complexity point of view, it is an interesting challenge for the future research to find out whether some results similar to those proved for scattered context grammars with erasing productions can also be achieved in case of nonerasing scattered context grammars.

\paragraph{Acknowledgements}
  Both authors have been supported by the Czech Ministry of Education under the research plan no. MSM~0021630528. The second author has also been supported by the Czech Grant Agency project no. 201/07/0005.

\bibliographystyle{eptcs}
\bibliography{dcfs2009}

\begin{thebibliography}{10}
\providecommand{\bibitemstart}[1]{\bibitem{#1}}
\providecommand{\bibitemend}{}
\providecommand{\bibliographystart}{}
\providecommand{\bibliographyend}{}
\providecommand{\url}[1]{\texttt{#1}}
\providecommand{\urlprefix}{Available at }
\providecommand{\bibinfo}[2]{#2}
\bibliographystart

\bibitemstart{Geffert:1988}
\bibinfo{author}{V.~Geffert} (\bibinfo{year}{1988}):
  \emph{\bibinfo{title}{Context-Free-Like Forms for the Phrase-Structure
  Grammars}}.
\newblock In: \bibinfo{editor}{M.~Chytil}, \bibinfo{editor}{L.~Janiga} \&
  \bibinfo{editor}{V.~Koubek}, editors: {\sl \bibinfo{booktitle}{MFCS}}, {\sl
  \bibinfo{series}{Lecture Notes in Computer Science}} \bibinfo{volume}{324}.
  \bibinfo{publisher}{Springer}, pp. \bibinfo{pages}{309--317}.
\bibitemend

\bibitemstart{GreHop}
\bibinfo{author}{S.~Greibach} \& \bibinfo{author}{J.~Hopcroft}
  (\bibinfo{year}{1969}): \emph{\bibinfo{title}{Scattered Context Grammars}}.
\newblock {\sl \bibinfo{journal}{Journal of Computer and System Sciences}}
  \bibinfo{volume}{3}, pp. \bibinfo{pages}{233--247}.
\bibitemend

\bibitemstart{kolar}
\bibinfo{author}{D.~Kol\'{a}\v{r}} (\bibinfo{year}{2008}):
  \emph{\bibinfo{title}{Scattered Context Grammars Parsers}}.
\newblock In: {\sl \bibinfo{booktitle}{Proceedings of the 14th International
  Congress of Cybernetics and Systems of WOCS}}. \bibinfo{publisher}{Wroclaw
  University of Technology}, pp. \bibinfo{pages}{491--500}.
\bibitemend

\bibitemstart{masopust:mono}
\bibinfo{author}{T.~Masopust} (\bibinfo{year}{2007}):
  \emph{\bibinfo{title}{Formal Models: Regulation and Reduction}}.
\newblock \bibinfo{type}{Ph.D. thesis}, \bibinfo{school}{Brno University of
  Technology, Faculty of Information Technology, Brno}.
\newblock \bibinfo{note}{On-line available at the author's web pages}.
\bibitemend

\bibitemstart{masopustTCS}
\bibinfo{author}{T.~Masopust} (\bibinfo{year}{2009}): \emph{\bibinfo{title}{On
  the Descriptional Complexity of Scattered Context Grammars}}.
\newblock {\sl \bibinfo{journal}{Theoretical Computer Science}}
  \bibinfo{volume}{410}(\bibinfo{number}{1}), pp. \bibinfo{pages}{108--112}.
\bibitemend

\bibitemstart{meduna:eatcs}
\bibinfo{author}{A.~Meduna} (\bibinfo{year}{1995}): \emph{\bibinfo{title}{A
  Trivial Method of Characterizing the Family of Recursively Enumerable
  Languages by Scattered Context Grammars}}.
\newblock In: {\sl \bibinfo{booktitle}{EATCS Bulletin}}.
  \bibinfo{publisher}{Springer-Verlag}, pp. \bibinfo{pages}{104--106}.
\bibitemend

\bibitemstart{Meduna00b}
\bibinfo{author}{A.~Meduna} (\bibinfo{year}{2000}):
  \emph{\bibinfo{title}{Generative Power of Three-Nonterminal Scattered Context
  Grammars}}.
\newblock {\sl \bibinfo{journal}{Theoretical Computer Science}}
  \bibinfo{volume}{246}, pp. \bibinfo{pages}{279--284}.
\bibitemend

\bibitemstart{meduna2}
\bibinfo{author}{A.~Meduna} (\bibinfo{year}{2000}):
  \emph{\bibinfo{title}{Terminating left-hand sides of scattered context
  productions}}.
\newblock {\sl \bibinfo{journal}{Theoretical Computer Science}}
  \bibinfo{volume}{237}, pp. \bibinfo{pages}{423--427}.
\bibitemend

\bibitemstart{techetAI}
\bibinfo{author}{A.~Meduna} \& \bibinfo{author}{J.~Techet}
  (\bibinfo{year}{2008}): \emph{\bibinfo{title}{Scattered Context Grammars that
  Erase Nonterminals in a Generalized k-Limited Way}}.
\newblock {\sl \bibinfo{journal}{Acta Informatica}}
  \bibinfo{volume}{45}(\bibinfo{number}{7}), pp. \bibinfo{pages}{593--608}.
\bibitemend

\bibitemstart{rychnov}
\bibinfo{author}{L.~Rychnovsk\'{y}} (\bibinfo{year}{2007}):
  \emph{\bibinfo{title}{Parsing of Context-Sensitive Languages}}.
\newblock In: {\sl \bibinfo{booktitle}{Proceedings of the 2nd Workshop on
  Formal Models, WFM 2007}}. \bibinfo{publisher}{Silesian University, Opava},
  pp. \bibinfo{pages}{219--226}.
\bibitemend

\bibitemstart{salomaa}
\bibinfo{author}{A.~Salomaa} (\bibinfo{year}{1973}):
  \emph{\bibinfo{title}{Formal languages}}.
\newblock \bibinfo{publisher}{Academic Press}, \bibinfo{address}{New York}.
\bibitemend

\bibitemstart{vaszil}
\bibinfo{author}{Gy. Vaszil} (\bibinfo{year}{2005}): \emph{\bibinfo{title}{On
  the descriptional complexity of some rewriting mechanisms regulated by
  context conditions}}.
\newblock {\sl \bibinfo{journal}{Theoretical Computer Science}}
  \bibinfo{volume}{330}, pp. \bibinfo{pages}{361--373}.
\bibitemend

\bibliographyend
\end{thebibliography}

\end{document}